\documentclass[12pt,a4paper]{article}
\usepackage[top=2.2cm,bottom=2.5cm,left=2.2cm,right=2.2cm]{geometry}
\usepackage{latexsym,amsmath,amsfonts,amssymb,amsthm,amscd,mathrsfs}

\newcommand{\be}{\begin{equation}}
\newcommand{\ee}{\end{equation}}
\newcommand{\bea}{\begin{eqnarray}}
\newcommand{\eea}{\end{eqnarray}}

\numberwithin{equation}{section}
\newcounter{thmcounter}
\numberwithin{thmcounter}{section}

\theoremstyle{definition}
\newtheorem*{acknowledgements}{Acknowledgements}

\newtheorem{remark}[thmcounter]{Remark}
\theoremstyle{plain}

\newtheorem{lemma}[thmcounter]{Lemma}
\newtheorem{proposition}[thmcounter]{Proposition}
\newtheorem{theorem}[thmcounter]{Theorem}

\def\1{{\boldsymbol 1}}                     %
\def\0{{\boldsymbol 0}}                     %
\def\cG{{\mathcal G}}                       %
\def\cH{{\mathcal H}}                       %
\def\cL{{\mathcal L}}                       %
\def\cP{{\mathcal P}}                       %
\def\cZ{{\mathcal Z}}                       %
\def\tr{\mathrm{tr}}                        %
\def\diag{\mathrm{diag}}                    %
\def\ri{{\rm i}}                            %
\def\C{\mathbb{C}}                          %
\def\N{\mathbb{N}}                          %
\def\R{\mathbb{R}}                          %
\def\T{\mathbb{T}}                          %
\def\SL{{\rm SL}}                           %
\def\UN{{\rm U}}                            %
\def\SU{{\rm SU}}                           %
\def\GL{{\rm GL}}                           %
\def\fH{\mathfrak{H}}                       %
\def\ds{\left.\frac{d}{ds}\right\vert_{s=0}}%
\def\cI{{\mathcal I}}
\def\cF{{\mathcal F}}
\def\cM{{\mathcal M}}
\def\cO{{\mathcal O}}
\def\cT{{\mathcal T}}                     %
\def\red{{\mathrm red}}
\def\cB{{\mathcal B}}
\def\fP{{\mathfrak{P}}}
\def\reg{\mathrm{reg}}
\def\red{\mathrm{red}}
\def\span{{\mathrm{span}}}
\def\Dr{\mathrm{Dress}}
\def\ad{\mathrm{ad}}
\def\dress{{\mathrm{dress}}}
\def\cV{{\mathcal V}}
\def\Ad{{\mathrm{Ad}}}
\def\id{{\mathrm{id}}}

\def\cQ{{\mathcal Z}}
\def\o{{\mathrm{o}}}

\begin{document}
\begin{center}
{\Large\bf
Poisson-Lie analogues of spin Sutherland models}
\end{center}

\medskip
\begin{center}
L.~Feh\'er${}^{a,b}$
\\

\bigskip
${}^a$Department of Theoretical Physics, University of Szeged\\
Tisza Lajos krt 84-86, H-6720 Szeged, Hungary\\
e-mail: lfeher@physx.u-szeged.hu

\medskip
${}^b$Department of Theoretical Physics, WIGNER RCP, RMKI\\
H-1525 Budapest, P.O.B.~49, Hungary\\
\end{center}

\medskip
\begin{abstract}
We present generalizations of the well-known trigonometric spin Sutherland models, which
were derived by Hamiltonian
reduction of `free motion' on cotangent bundles of compact simple Lie groups based on the conjugation action.
Our models result by reducing the corresponding  Heisenberg doubles
with the aid of a Poisson-Lie analogue of the conjugation action.
We describe the reduced symplectic structure
and show that the `reduced main Hamiltonians' reproduce the spin Sutherland model by keeping only
their leading terms.
The solutions of the equations of motion emerge from geodesics
on the compact Lie group via the standard projection method and possess many first integrals.
Similar hyperbolic spin Ruijsenaars--Schneider type models were obtained previously by L.-C.~Li
using a different method, based on coboundary
dynamical Poisson groupoids,
 but their relation with spin Sutherland models was not discussed.
 \end{abstract}

{\linespread{0.8}\tableofcontents}

\newpage
\section{Introduction}
\label{sec:1}

Integrable systems of particles moving in one dimension have been studied intensively
for nearly 50 years, beginning with the pioneering papers of Calogero \cite{Cal}, Sutherland \cite{S} and Moser \cite{M}.
Thanks to their fascinating mathematics and diverse applications \cite{E,N,OP,RBanff,vDV},
the interest in these models shows no sign of diminishing.  New connections to
mathematics and new applications are still coming to light in the current
literature, see e.g.~\cite{C,CF1,ILLS,KLOZ,SS,Zab}.

The richness of these models is also due to their many generalizations and deformations.
These are associated with different interaction potentials (from rational to elliptic), root systems
and extensions with internal degrees of freedom.
We call `Sutherland models' the systems defined by trigonometric or hyperbolic potentials.
For all these systems, classical and
quantum mechanical versions are studied separately, and one  needs to pay attention to the distinct features
of the
systems with real particle positions and their complexifications.
The investigations of  Ruijsenaars--Schneider (RS) type deformations \cite{RBanff,RS}
is motivated, for example, by relations to solitons, spin chains,
special functions and double affine Hecke algebras.

The internal degrees of freedom are colloquially called `spin', and can be
of two rather different kinds. First, the point particles can carry spins varying in a vector space,
as is the case for the Gibbons--Hermsen  models \cite{GH} and their RS type generalizations
introduced by Krichever and Zabrodin \cite{KZ}. Second, the models can involve a collective spin variable
that typically belongs to a coadjoint orbit, and is not
assigned separately to the particles.
An example of this second type
is the trigonometric spin Sutherland model
defined classically by a Hamiltonian of the following form:
\be
\cH_{\mathrm{Suth}}(e^{\ri q}, p,\xi) =  \frac{1}{2} \langle p, p \rangle +
\frac{1}{2} \sum_{\alpha > 0} \frac{1}{\vert \alpha \vert^2}
\frac{\vert \xi_\alpha \vert^2}{\sin^2 \frac{\alpha(q)}{2}}.
\label{1}\ee
Here, $\langle\ ,\ \rangle$ is the Killing form of
the complexification of the Lie algebra $\cG$ of a compact simple Lie group, $G$,
$e^{\ri q}$ belongs to the interior, $\T^o$, of a Weyl alcove\footnote{A Weyl alcove is a fundamental domain
for the conjugation action of $G$ on itself.}
in the maximal torus $\T < G$, and $\ri p$ varies in the Lie algebra $\cT$ of $\T$.
The spin variable  $\xi = \sum_{\alpha>0} \left( \xi_\alpha E_\alpha -  \xi_\alpha^*E_{-\alpha}\right)$ lies in
$\cO_0:= \cO \cap \cT^\perp$, where $\cO$ is an arbitrarily chosen coadjoint orbit of $G$, and
$\alpha$ runs over the positive roots.
More precisely, the  Hamiltonian $\cH_{\mathrm{Suth}}$ lives on the phase space
$T^* \T^o \times (\cO_0/\T)$.
These spin Sutherland models can be interpreted as Hamiltonian reductions
of free motion on $G$, relying on the cotangent lift of the conjugation action
of $G$ on itself.
The reduction can be utilized to show their integrability, and to analyze their quantum mechanics
with the aid of representation theory \cite{EFK,FP1,FP2,Res1,Res2,Res3}.
Spinless models can be obtained in this way only for $G = \SU(n)$, using a minimal
coadjoint orbit, for which the $\T$-action on $\cO_0$ is transitive.

As was shown by Li and Xu \cite{LiXu}, the models (\ref{1}) (and generalizations) result from a different construction as well.
Their construction is built on  Lie algebroids defined using
the solutions of the classical dynamical Yang-Baxter equation.  For the connection of these approaches,
 we refer to \cite{FP1}.

Our original motivation for the present work stems from \cite{FK1}, where it was shown how
the Ruijsenaars--Schneider deformation of the standard spinless Sutherland model arises
from a Hamiltonian reduction of the Poisson-Lie counterpart of $T^*\SU(n)$,
the so-called Heisenberg double. To obtain the spinless model, one has to choose
a minimal dressing orbit of $\SU(n)$ in setting up the reduction.
It is natural to expect that the analogous reduction of the Heisenberg double
of any compact simple Lie group, along an arbitrary dressing orbit, will lead
to a  generalization of the spin Sutherland model (\ref{1}).
Motivated by the recent interest in spin Calogero--Moser and RS models \cite{AO,CF2,ILLS,KLOZ,Res2,Res3,SS},
we  take up this issue here.

In fact, the purpose of this paper is to describe
the spin RS type models that descend from the Heisenberg double of a
compact simple Lie group $G$.
The so-called main reduced Hamiltonians, which originate from the characters of the complexification of $G$,
 will turn out to have $\cH_{\mathrm{Suth}}$ (\ref{1})
as their leading term, up to cubic
 and higher order terms
in $p$ and a spin variable.
(Here and throughout the paper,
we refer to the total, or combined,  degree in $p$ and the spin variable.
For example, $p^3$, $p^2\sigma$, $p\sigma^2$ and $\sigma^3$ all have degree $3$.)
The spin variable now belongs to a reduced dressing orbit of the Poisson-Lie group $G$.
The dressing orbits are the Poisson-Lie analogues of the
coadjoint orbits, and in the compact case each dressing orbit is diffeomorphic,
and is even symplectomorphic \cite{GW},
to a coadjoint orbit.

In the $\SU(n)$ case, using analytic continuation from trigonometric to hyperbolic
functions,
our models reproduce the spin RS type
equations of motion derived by Braden and Hone \cite{BH} from the
soliton solutions of $A_{n-1}$ affine Toda theory with imaginary coupling.
These equations of motion were interpreted previously
by L.-C.~Li \cite{Li1,Li2} as examples of  spin RS type Hamiltonian systems obtained
by applying (discrete and Hamiltonian) reductions to the coboundary dynamical Poisson
groupoids that underlie the geometric interpretation of the classical dynamical Yang-Baxter equation \cite{EV}.
Remembering also the alternative constructions of spin  Sutherland models
\cite{FP1,LiXu},
it is clear that there must exist a connection
between our systems and corresponding systems of \cite{Li1,Li2}.
The two approaches  are substantially different, but the analytic continuation
of our models appears to yield a subclass of those in \cite{Li2}.
This is discussed further in Remark 6.3 and in Section 7, together with other approaches to spin RS type models.
The precise connection will be explored in detail in a subsequent publication.

In the trigonometric/hyperbolic case, the papers \cite{AO,CF2} contain two different reduction
 treatments of the
 holomorphic spin RS systems of Krichever and Zabrodin \cite{KZ} associated with $\SL(n,\C)$.
The complexifications of our systems associated with $\SU(n)$ do not reproduce those systems.
The systems studied in \cite{AO,CF2,KZ} feature
individual spins attached to the particles, while our systems involve only collective spin variables.
 Hamiltonian reductions leading directly to distinct real forms of the
complex trigonometric/hyperbolic spin RS systems of \cite{KZ},  as well as the elliptic systems
and the case of general root systems,
should be developed in the future.

Now we sketch the organization of the rest of the paper.
We start in Section 2 by recalling the reduction treatment of the spin Sutherland models,
which can be found in many sources (see e.g.~\cite{FP1,Res1}). This section puts our generalization in context,
and provides motivation for it. In Section 3, we present the rudiments of the
standard Heisenberg double of a compact Poisson-Lie group
 and its `natural free system' that we shall reduce.
To our knowledge, this `free system' first appeared in \cite{Zakfirst}, and was utilized previously, for example,  in
\cite{FK1, FK2}.
Then,  in Section 4,  we describe the structure of the reduced phase space
in an as complete manner as is known for the spin Sutherland models.
In Section 5, we show that the spin Sutherland Hamiltonian (\ref{1}) is recovered as the leading term
of the reduced main Hamiltonians associated with the characters of the finite dimensional
irreducible representations of $G^\C$.
In Section 6, we develop the form of the reduced Hamilton equations, detail the
projection approach for constructing their solutions, and display a large number of
integrals of motion. We also present spectral parameter dependent Lax equations.
Section 7 contains an outlook on future studies and open questions.

Finally, let us highlight our main results.
The first important result is the description of the symplectic structure on
a dense open submanifold of the reduced phase space, given by Theorem 4.3.
The second significant result, presented in Section 5.1, is the characterization of the reduced main
Hamiltonians from which we can recover the spin Sutherland Hamiltonian (\ref{1}) as
the leading term.  More precisely, we also recover the Lax matrix of the model (\ref{1})
as a suitable leading term,
and explain in  Remark 5.1 how our models can be viewed as
one-parameter deformations of the spin Sutherland models.
In the $\SU(n)$ case, we obtain an explicit solution of the moment map
constraints, see Section 5.2.
Further results are the simple derivation of the reduced equations of motion
and their solutions in Section 6.1, and the arguments put forward in
Section 6.2 that indicate
their integrability.

\section{Spin Sutherland model from reduction}

First of all, we fix the Lie theoretic \cite{K} notations  that will be used throughout the paper.
Let $\cG^\C$ be a complex simple Lie algebra equipped with the normalized Killing form
$\langle\ ,\ \rangle$, and a Chevalley basis given by
$E_{\pm \alpha}$ $(\alpha\in \Phi^+)$ and $ T_{\alpha_k}$ $(\alpha_k\in \Delta)$, where
$\Phi^+$ and $\Delta$ denote the sets of positive and simple roots, respectively.  The normalization is such that the long
roots have length $\sqrt{2}$ and
$\langle E_\alpha, E_\beta \rangle= \frac{2}{\vert \alpha \vert^2} \delta_{\alpha, - \beta}$ holds.
We let $N:= \operatorname{dim}_\C (\cG^\C)$ and write $\cG^\C_\R$ for $\cG^\C$ regarded as a Lie algebra
over the real numbers. We then have the real vector space direct sum
\be
\cG^\C_\R = \cG + \cB,
\label{1*}\ee
where
\be
\cG=\span_\R\{ (E_\alpha - E_{-\alpha}), \ri (E_\alpha + E_{-\alpha}), \ri T_{\alpha_k}\mid
\alpha\in \Phi^+, \alpha_k\in \Delta\}
\label{H2}\ee
is the compact real form of $\cG^\C$
and
\be
\cB=\span_\R\{ E_\alpha, \ri E_\alpha,  T_{\alpha_k}\mid
\alpha\in \Phi^+, \alpha_k\in \Delta\}
\label{H3}\ee
is a `Borel' subalgebra.
Consider the connected and simply connected complex Lie group,  $G^\C$,  associated with $\cG^\C$.
When viewed as a real Lie group, we denote it as $G_\R^\C$, and let
$G$ and $B$ stand for the connected Lie subgroups of $G_\R^\C$ corresponding to the subalgebras
$\cG$ and $\cB$, respectively.
The restriction of $\langle\ ,\ \rangle$ to $\cG$ is the negative definite Killing form of $\cG$.
The subalgebras $\cG$ and $\cB$ of $\cG_\R^\C$ are isotropic with respect to the non-degenerate
invariant bilinear form on $\cG_\R^\C$ provided by the imaginary part of the complex Killing form,
which we denote as
\be
(X,Y):= \mathrm{Im} \langle X, Y \rangle,\qquad \forall X,Y\in \cG^\C_\R.
\label{H4}\ee
Notationwise, we shall `pretend' that we are always dealing with matrix Lie groups.
For example, the left-invariant Maurer--Cartan form on $G$ will be written as $g^{-1} dg$.
If desired, our matrix Lie group notations can be easily converted into more abstract symbolism.

Now, we briefly summarize the reduction that we shall generalize.
We start with the master phase space $M:= T^* G \times \cO$,
where $T^*G$ is the cotangent bundle, and $\cO$ is a coadjoint orbit of the  Lie group $G$.
The phase space is endowed with the  Poisson maps
\be
 J_L: M \to \cG^*, \quad J_R: M \to \cG^*, \quad J_\cO: M \to \cG^*,
\label{I1}\ee
where $J_L$ ($J_R$) generates the Hamiltonian left-action of $G$  on $T^*G$ engendered by
the left-shifts (right-shifts) and $J_\cO$ is obtained by combining
projection to $\cO$ with the
 tautological embedding of $\cO$ into $\cG^*$. One then considers  the moment map
\be
\mu:= J_L + J_R + J_\cO
\label{I2}\ee
that generates the `conjugation action' of $G$ on $M$.
A dense open subset of the reduced  phase space belonging to the zero value of $\mu$ can be
identified with the (stratified) symplectic space (see \cite{SL,Sn})
\be
M_\red^\reg = T^* \T^o \times \cO_0/\T,
\label{27}\ee
where $\T^o$ is the interior  of a Weyl alcove in the maximal torus $\T < G$, and $\cO_0/\T$
is the symplectic reduction of $\cO$ by $\T$ at the zero value of the respective moment map.

Next, we explain how the above description of the reduced phase space comes about.
For this, we let $\pi_G: T^*G \to G$ denote the bundle projection and
use the diffeomorphism
\be
(\pi_G, J_R): T^* G \to G \times \cG^*.
\ee
Together with the identification $\cG^* \simeq \cG$ defined by the Killing form of $\cG$,
this allows us to take
\be
G \times \cG \times \cO = \{ (g,J,\xi)\}
\ee
as the model of $M$.
Then the symplectic form $\omega$  of $M$ can be written as
\be
\omega = - d \langle J, g^{-1} d g \rangle + \omega_{\cO},
\ee
where $\omega_\cO$ is the canonical symplectic form of $\cO$.
The subset of $M$ on which $\mu=0$ holds is specified by
the constraint equation
\be
J - g J g^{-1} + \xi =0.
\label{11}\ee
We can bring\footnote{The inverse is used since  $g^{-1}$ is
the counterpart of $g_R$ that will appear later.}
 $g^{-1}$ into its representative
$Q\in \T$, which we parametrize as
\be
Q = \exp(\ri q).
\ee
Then the constraint (\ref{11}) becomes
\be
e^{-\ri q} J e^{\ri q} - J = \xi.
\ee
We assume that $q$ is regular, i.e.~$e^{\ri q}$ belongs to the interior $\T^o$ of a Weyl alcove,
which permits us to solve the moment map constraint as follows:
\be
J = -\ri p + \sum_{\alpha\in \Phi^+}  (J_\alpha E_\alpha - J_\alpha^* E_{-\alpha}),
\quad
\xi =  \sum_{\alpha\in \Phi^+}  (\xi_\alpha E_\alpha - \xi_\alpha^* E_{-\alpha}),
\label{214}\ee
where $\ri p\in \cT$ is arbitrary and
\be
J_\alpha = \frac{\xi_\alpha}{ e^{- \ri \alpha(q)} - 1}.
\label{215}\ee
In this way, we obtained a `partial gauge fixing' parametrized by
\be
\T^o \times \cT \times \cO_0 = \{ (e^{\ri q}, \ri p, \xi)\}.
\ee
We still need to divide this gauge slice by the residual gauge transformations,
generated by $\T$, which act only on $\cO_0$.
This yields the model (\ref{27}) of the reduced phase space, where $T^* \T^o$
is identified with
$\T^o \times \cT$.
The reduced symplectic structure can be displayed as
\be
\omega_\red =   \langle d p \stackrel{\wedge}{,} dq \rangle + \omega_{\cO}^\red.
\ee
Here,  $\omega_\cO^\red$ stands for
the (stratified) symplectic structure arising  form $(\cO, \omega_\cO)$,  reduced by
the $\T$-action at zero moment map value.
That is,  $\omega_\cO^\red$ encodes the restriction of the Poisson brackets of the
elements of $C^\infty(\cO)^\T$  to $\cO_0 = \cO\cap \cT^\perp$.

Upon substitution of (\ref{215}), the `free' Hamiltonian
\be
\cH(g,J,\xi) := - \frac{1}{2} \langle J, J \rangle
\ee
yields
the spin Sutherland Hamiltonian $\cH_{\mathrm{Suth}}$ given by equation (1.1).
The flow generated by $\cH$ is called  `free motion':
\be
g(t) = g(0) \exp( t J(0) ), \quad J(t) = J(0), \quad \xi(t) = \xi(0),
\label{218}\ee
and the dynamics governed by $\cH_{\mathrm{Suth}}$ results by projecting this to
the reduced phase space.
The Hamiltonian $\cH$ is a member of the Abelian Poisson algebra
\be
C_I(M) := J_R^* (C^\infty(\cG^*)^G),
\label{I6}\ee
whose functional dimension equals $r = \mathrm{rank}(G)$.
The elements of $C_I(M)$ Poisson commute with all elements of
the Poisson algebra $C_J(M)$ generated
by the components of $J_L, J_R$ and $J_\cO$. The functional dimension of the
`algebra of integrals of motion' $C_J(M)$ is
$\dim(M) - r$, since the functions
of $J_L$ and $J_R$ are connected by $r$ independent relations,
which express the equality $f \circ J_R = f \circ (-J_L)$ for every $f\in C^\infty(\cG^*)^G$.
This means \cite{MF,Nekh}  that the free Hamiltonians $C_I(M)$, with their integrals of motion $C_J(M)$,
represent a degenerately integrable (in other words
non-commutative integrable or super-integrable) system on $M$.
The various notions of integrability and their relations are reviewed,
for example, in \cite{J,Zung2}.

All elements of $C_I(M)$ descend to smooth functions on the reduced phase space.
Their reduced flows can be found via the projection method, similarly to the case of $\cH$,
and all those flows are complete on the full reduced phase space, $M_\red=\mu^{-1}(0)/G$.
It was shown by Reshetikhin \cite{Res1,Res2,Res3} that the degenerate integrability of the
free Hamiltonian $\cH$ (\ref{1}) is inherited at the reduced level with
analytic integrals of motion, at least for generic coadjoint orbits and on a dense open
subset of $M_\red$.
Liouville integrability in the same generic case follows from the results of \cite{LN}.
It would require further work to obtain a full understanding for arbitrary orbits and arbitrary
symplectic strata \cite{SL,Sn} of $M_\red$.
We do not go into this intricate issue, but wish  to display a large number
of integrals of motion that survive the reduction.
Namely, let $\cP(J, gJg^{-1})$
be an \emph{arbitrary polynomial} in its non-commutative variables (viewed as elements of
the enveloping algebra).
Then evaluate the trace of this polynomial in an arbitrary finite dimensional unitary
representation $\rho$ of $G$.
It  is easy to see that all the  functions $\tr_\rho\left( \cP(J, g J g^{-1}) \right)$
Poisson commute with every element of $C_I(M)$ and they are $G$-invariant with respect
to the conjugation action.  We suspect that the resulting integrals of motion
are sufficient for the integrability of the reduction of $C_I(M)$ in general.

Later we shall derive spin RS type systems, which will be compared to
the spin Sutherland systems. Instead of the identification
$\cG^* \simeq \cG$, the comparison will be done  using
another model of $\cG^*$.
 This model is defined
by realizing any linear functional $\phi$  on $\cG$ in the form
$\phi(X)  = (\tilde \xi, X), \quad \forall X \in \cG$,
 where $\tilde \xi$ is from the subalgebra $\cB$ of $\cG^\C_\R$.
 The two models of  $\cG^*$, $\cG$ and $\cB$,  are in bijection via the equality
 \be
 \phi(X)  = \langle \xi, X\rangle =  (\tilde \xi, X), \qquad \xi\in \cG,\, \tilde \xi\in \cB.
 \ee
 This implies that
 \be
 \xi = \sum_{k=1}^r \ri \xi^k T_{\alpha_k} + \sum_{\alpha \in \Phi^+} \left(\xi_\alpha E_\alpha  - \xi_\alpha^* E_{-\alpha}\right)
 \ee
 corresponds to
 \be
 \tilde \xi = \sum_{k=1}^r \tilde\xi^k T_{\alpha_k} +  \sum_{\alpha \in \Phi^+} \tilde \xi_\alpha E_\alpha
 \quad\hbox{with}\quad \tilde \xi_\alpha = - 2\ri \xi_\alpha
 \quad\hbox{and}\quad \tilde \xi^k = - \xi^k.
 \label{tildevar}\ee
 Thus, the spin Sutherland Hamiltonian (\ref{1}) can be casted as
 \be
\cH_{\mathrm{Suth}}(e^{\ri q},p,\tilde \xi) =  \frac{1}{2} \langle  p,  p \rangle +
\frac{1}{8} \sum_{\alpha \in \Phi^+} \frac{1}{\vert \alpha \vert^2}
\frac{\vert \tilde \xi_\alpha \vert^2}{\sin^2 \frac{\alpha(q)}{2}}.
\label{tilde1}\ee

\section{Unreduced free system on the Heisenberg double}

In what follows we freely use basic notions and results from the theory of
Poisson-Lie groups,
as can be found, e.g., in the reviews \cite{CP,KS, STS2}.
One may also consult \cite{FK2}, where similar background material
as given below is described in more detail.

We start by noting that
the Lie algebra $\cG_\R^\C$ and its subalgebras $\cG$ and $\cB$ form a Manin triple.
Consequently,  $G$ and $B$ are Poisson-Lie groups in duality.  The multiplicative Poisson bracket on
$C^\infty(G)$ is given by
\be
\{ \phi_1, \phi_2\}_G(g)=  \left(  g^{-1}(d^L \phi_1(g)) g, d^R\phi_2(g)\right),
\qquad
\forall \phi_1, \phi_2\in C^\infty(G),\, \forall g\in G,
\label{H5}\ee
and that on $C^\infty(B)$ is given by
\be
\{ f_1, f_2\}_B(b)= -  \left( b^{-1}(d^L f_1(b)) b , d^Rf_2(b)\right),
\qquad
\forall f_1, f_2\in C^\infty(B),\, \forall b\in B.
\label{H6}\ee
Here,
for a real function $\phi \in C^\infty(G)$  the left and right derivatives
$d^{L,R}\phi \in C^\infty(G,\cB)$ are defined by
\be
\ds \phi(e^{sX} g e^{sY}) = \left(X, d^L \phi(g)\right) +
\left(Y, d^R \phi(g)\right),
\qquad
\forall X,Y\in \cG, \, \forall g\in G,
\label{H7}\ee
and
$d^{L,R}f \in C^\infty(B,\cG)$ for a real function $f\in C^\infty(B)$  are defined similarly.
In the above Poisson bracket formulas conjugation is an informal shorthand for the adjoint action
of $G^\C_\R$ on its Lie algebra.

The manifold $G_\R^\C$ carries  a natural symplectic structure, $\Omega_+$, which
goes back to Semenov-Tian-Shansky \cite{STS1} and to Alekseev and Malkin \cite{AM}.
When equipped with $\Omega_+$, $G_\R^\C$ is a Poisson-Lie analogue of the cotangent bundle
$T^*G$, alias the `Heisenberg double' of the Poisson-Lie group $G$.
To present $\Omega_+$, let us recall that every element $K \in G_\R^\C$ admits the alternative Iwasawa decompositions
\be
K = b_L g_R^{-1} = g_L b_R^{-1},\qquad b_L, b_R\in B,\, g_L,g_R\in G,
\label{H8}\ee
that define diffeomorphisms between $G_\R^\C$ and $G \times B$.
The pair $g_R, b_R$ or the pair $g_L, b_L$ can be also used as free variables in $G\times B$, utilizing
the relation
\be
g_L^{-1} b_L = b_R^{-1} g_R.
\ee
By making use of these decompositions, we have
\be
\Omega_+ =\frac{1}{2} \left(d b_L b_L^{-1}\stackrel{\wedge}{,} d g_L g_L^{-1} \right)+
  \frac{1}{2}\left(d b_R b_R^{-1}\stackrel{\wedge}{,} d g_R g_R^{-1}\right).
\label{H9}\ee
It is useful to introduce the maps $\Lambda_L, \Lambda_R$ from $G_\R^\C$ to $B$
and the maps $\Xi_L, \Xi_R$ from $G_\R^\C$ to $G$ by setting
\be
\Lambda_L(K):= b_L, \quad \Lambda_R(K) := b_R, \quad \Xi_L(K):=g_L, \quad \Xi_R(K):= g_R.
\label{H10}\ee
These are Poisson maps with respect to the Poisson structure associated with $\Omega_+$ and
the multiplicative Poisson structures on $B$ and on $G$, respectively.

The group $G$ acts on $B$ by the (left) dressing action
given by
\be
\Dr_\eta(b) = \Lambda_L(\eta b),\qquad \forall \eta\in G,\,b \in B,
\label{H11}\ee
which is a Poisson action.  The induced infinitesimal  action of $\cG$ on $B$ reads
\be
\mathrm{dress}_X b = b (b^{-1} X b)_\cB,\quad  \forall X\in \cG,
\label{dress}\ee
where on the right-hand side we use projection along $\cG$, by means of (\ref{1*}).
The ring of invariants $C^\infty(B)^G$ forms the center of
the Poisson algebra of $B$. Thus we obtain an algebra of commuting `free Hamiltonians', $C_I(G_\R^\C)$,
 by the definition
\be
C_I(G_\R^\C) := \Lambda_R^*( C^\infty(B)^G).
\label{H12}\ee
It is worth remarking that $\Lambda_R^*(C^\infty(B)^G) =\Lambda_L^*(C^\infty(B)^G)$.
The flow generated by any Hamiltonian $\Lambda_R^*(h) \in C_I$ can be written down explicitly:
\be
g_R(t) = \exp\left[t d^L h(b_R(0)) \right] g_R(0),
\qquad
b_L(t)=b_L(0), \quad b_R(t) = b_R(0).
\label{H13}\ee
Notice the  similarity with the corresponding flow\footnote{As was noted before,
$g$ in (\ref{218})  corresponds to $g_R^{-1}$.
The analogous eq.~(2.38) in \cite{FK2} contains a typo.}
(\ref{218}) on $T^*G$.
These Hamiltonians Poisson commute with all the elements
of $\Lambda_L^*(C^\infty(B))$ and $\Lambda_R^*(C^\infty(B))$, which together generate
the Poisson algebra of the integrals of motion, denoted as $C_J(G_\R^\C)$.
The functional dimension of $C_I$ is the rank $r$ of $\cG^\C$, while the functional dimension
of $C_J$ is $(2N-r)$. The latter statement follows since for any $f\in C^\infty(B)^G$ we have
\be
\Lambda_L^*(f) = \Lambda_R^* ( f \circ \mathrm{inv}_B),
\label{H14}\ee
where $\mathrm{inv}_B$ is the inversion map on the group $B$.
These identities represent $r$ independent relations between $\Lambda_L^*(C^\infty(B))$ and
$\Lambda_R^*(C^\infty(B))$, which otherwise give independent functions.
Consequently \cite{J,Res2,Zung2}, the  Hamiltonians in $C_I$ (\ref{H12}) define
a degenerate integrable system.

The following  model of the Poisson manifold $B$ is often useful.
Let
\be
\fP:= \exp(\ri \cG)
\label{H15}\ee
denote the closed submanifold of $G_\R^\C$
diffeomorphic to $\ri \cG$ by the exponential map.  Note that $\cG$ and $G$ are pointwise fixed by
corresponding Cartan involutions \cite{K} $\theta$ and $\Theta$ of $\cG_\R^\C$ and $G_\R^\C$.
Somewhat colloquially, we write
\be
X^\dagger := - \theta(X), \quad K^\dagger := \Theta(K^{-1})
\quad\hbox{for}\quad X\in \cG^\C_\R,\,\, K\in G_\R^\C,
\label{H16}\ee
since this anti-involution can be arranged to be the usual matrix adjoint for the classical groups.
Then the map
\be
m: B \to \fP,
\qquad
B\ni b\mapsto  bb^\dagger \in \fP
\label{H17}\ee
is a diffeomorphism,
which converts the dressing action of $G$ on $B$ into the conjugation action of $G$ on $\fP$. That is, we have
\be
m \circ \Dr_\eta = C_\eta \circ m \qquad \hbox{where}\qquad  C_\eta(P):= \eta P \eta^{-1} \quad\forall P\in \fP.
\label{H18}\ee
It follows that any dressing orbit, $\cO_B$, is diffeomorphic by $m$ to
$\cO_\fP = m(\cO_B)$, and one can parametrize it as
\be
\cO_\fP=m(\cO_B) = \{ \exp(2 \ri X)\mid X\in \cO_\cG \},
\label{H19}\ee
where $\cO_\cG$ is an adjoint orbit of $G$.
In terms of this exponential parametrization,
the form of the Poisson structure on  $\fP=m(B)$ is described in \cite{FM1}.

We end this section  by recording another useful feature of the Poisson structure on $B$.
For this, let us consider the decompositions
\be
\cB = \cB_0 + \cB_+,\qquad \cG = \cT + \cT^\perp
\ee
where $\cB_0$ (resp. $\cB_+$) is spanned by Cartan elements (resp. root vectors).
Choose an arbitrary basis $\{X^\alpha\}$ of $\cT$ and a basis $\{Y^i\}$ of $\cT^\perp$.
Every element $b\in B$ can be uniquely written in the form
\be
b = b_0 b_+ =e^{\beta_0} e^{\beta_+} \quad\hbox{with}\quad \beta_0 \in \cB_0,\, \beta_+ \in \cB_+,
\ee
and the components
\be
\beta_0^\alpha := (\beta_0, X^\alpha), \quad \beta_+^i := (\beta_+, Y^i)
\ee
can be taken as coordinate functions on $B$.
The Poisson brackets of these functions satisfy
\be
\{ \beta_0^\alpha, \beta_0^\gamma\}_B =0,
\quad
\{ \beta_+^k, \beta_0^\gamma \}_B = ( [Y^k, X^\gamma], \beta_+)
\label{B0}\ee
and
\be
\{\beta_+^i, \beta_+^j\}_B = ([Y^i, Y^j], \beta_0 + \beta_+) + \mathrm{o}(\beta_+, \beta_0).
\label{B+}\ee
The Poisson brackets $\{\beta_+^i, \beta_+^j\}_B$ are polynomials
in $\beta_+$ and trigonometric polynomials in $\beta_0$.
Equation (\ref{B+}), where $\mathrm{o}(\beta_+, \beta_0)$ stands for terms whose combined degree in
the components of $\beta_+$ and $\beta_0$ is at least 2,
 shows that the linear part of the Poisson brackets of the variables $\beta_0^\alpha, \beta_+^i$ is the Lie-Poisson
bracket of $\cG$.

\section{Reduction along an arbitrary dressing orbit}

We recall that the dressing orbits $\cO_B$
are the symplectic leaves in $B$, and let $\Omega_{\cO_B}$
stand for the symplectic form on $\cO_B$.
Before defining the reduction,
we extend the phase space $G^\C_\R$ by a non-trivial dressing orbit, i.e., we consider the
unreduced phase space
\be
\cM:= G_\R^\C \times \cO_B = \{ (K, S)\mid K\in G_\R^\C,\,\, S\in \cO_B\}
\label{N1}\ee
equipped with the symplectic form
\be
\Omega = \Omega_+ + \Omega_{\cO_B}.
\label{N2}\ee
The Abelian Poisson algebra (\ref{H12}) is trivially extended to yield $C_I(\cM)$, whose
elements do not depend on $S\in \cO_B$,  and the algebra of the integrals of motion $C_J(G_\R^\C)$
is extended to
\be
C_J(\cM)= (\Lambda_L,\Lambda_R, \Lambda_{\cO_B})^*\left(C^\infty(B \times B \times \cO_B)\right),
\label{N3}\ee
where $\Lambda_{\cO_B}$ is the obvious projection from $\cM$ to  $\cO_B \subset B$, and $\Lambda_L, \Lambda_R$
(\ref{H10}) are regarded as maps from $\cM$ to $B$.
That is, $C_J(\cM)$ contains all functions of $b_L, b_R$ (\ref{H8}) and $S$.
This extension maintains the degenerate integrability.

We shall study Marsden--Weinstein type reduction \cite{Lu} at the unit value $e\in B$
of a suitable Poisson-Lie moment  map $\Lambda: \cM\to B$.
Concretely, we introduce the map $\Lambda$ by taking the product
\be
\Lambda = \Lambda_L \Lambda_R \Lambda_{\cO_B},
\label{N4}\ee
i.e.,
\be
\Lambda(K,S) = \Lambda_L(K) \Lambda_R(K) S.
\label{N5}\ee
Clearly, the product is a proper generalization of the sum in (\ref{I2}).
This definition gives a Poisson map because the 3 factors of $\Lambda$ are Poisson maps into $B$
and they pairwise Poisson commute.
We know from general theory \cite{Lu} that the Poisson map $\Lambda$ generates
an infinitesimal left-action of $\cG$ on
$\cM$. Namely, the vector field $X_\cM$ on $\cM$ corresponding to $X\in \cG$ operates on
$f\in C^\infty(\cM)$ by the following formula:
\be
df(X_\cM) = (X, \{ f, \Lambda\}_\cM \Lambda^{-1} ),
\label{N6}\ee
where $\{\ ,\ \}_\cM$ is the Poisson bracket on functions on $\cM$, and notationwise
we pretend that $B$ is a matrix Lie group.
The $\cG$-action (\ref{N6})
integrates to a global Poisson-Lie action of $G$ on $\cM$, denoted below $\Psi: G\times \cM \to \cM$.

\begin{lemma}
\label{lemma:4.1}
The action of $\eta\in G$ on $\cM$ is given by the following diffeomorphism
$\Psi_\eta$,
\be
\Psi_\eta(K, S) = (\eta K \Xi_R(\eta b_L), \Dr_{\Xi_R( \eta b_L b_R)^{-1}}(S)),
\label{N7}\ee
where we use the notations introduced in (\ref{H8}), (\ref{H10}) and  (\ref{H11}). The map
\be
\Psi: G \times \cM \to \cM, \quad \Psi(\eta, K,S) = \Psi_\eta(K,S)
\label{N8}\ee
is Poisson, and the moment map $\Lambda$ is equivariant: $\Lambda \circ \Psi_\eta = \Dr_\eta \circ \Lambda$.
\end{lemma}

\begin{proof}
One can verify that this formula defines a group action, and the induced infinitesimal action
reproduces the derivations given by the moment map according to (\ref{N6}).
\end{proof}

\begin{remark}
\label{rem:4.2}
One can check that $\Lambda_L$ and $\Lambda_L \Lambda_R$ are also equivariant in the sense that
\be
\Lambda_L \circ \Psi_\eta = \Dr_\eta \circ \Lambda_L,\qquad (\Lambda_L \Lambda_R) \circ \Psi_\eta =
\Dr_\eta \circ (\Lambda_L \Lambda_R).
\label{N9}\ee
It follows that all elements of $\Lambda_L^*(C^\infty(B)^G) = \Lambda_R^*(C^\infty(B)^G):= C_I(\cM)$
 are invariant with respect to $\Psi_\eta$.
Without the extension of the Heisenberg double by the
dressing orbit, the action (\ref{N7}) was introduced in \cite{Kli},
where it was called `quasi-adjoint action'.
\end{remark}

Now, we are interested in the reduced phase space
\be
\cM_\red := \Lambda^{-1}(e)/G.
\label{N10}\ee
For certain orbits $\cO_B$  this is a smooth symplectic manifold. In general,
it is a union of smooth symplectic manifolds of various dimension, a so-called
stratified symplectic space \cite{SL,Sn}. Its structure turns out to be quite similar to what
occurs in the cotangent bundle case. In particular, a reduction of the orbit $\cO_B$
itself will come to fore shortly in our description.

The maximal torus $\T <G$ is
 a Poisson-Lie subgroup of $G$, on which the Poisson structure vanishes.
Hence the dressing Poisson action of $G$ on $\cO_B$ restricts to an ordinary
Hamiltonian action of $\T$. This action operates simply by conjugation.
Writing $S\in \cO_B$ in the form
\be
S= S_0 S_+ \quad\hbox{with}\quad S_0\in B_0,\,S_+\in B_+,
\label{N12}\ee
the map
\be
S \mapsto \log (S_0) \in \cB_0
\label{N14}\ee
is the moment map for the action of $\T$ on $\cO_B$, as follows, for example, from (\ref{B0}).
Here, $\cB_0$ plays the role of the dual space of $\cT$, via the bilinear form (\ref{H4}).
By setting this moment map to zero, i.e. setting $S_0$ equal to the unit element,
we obtain the reduced dressing orbit
\be
\cO_B^\red = \{ S_+ \in \cO_B\}/ \T,
\label{N15}\ee
which itself is a stratified symplectic space.

Let $G^\reg\subset G$ be the set of regular elements.
The space of the conjugacy classes in $G^\reg$ is a smooth manifold, which
can be identified with an open Weyl alcove $\T^o$, i.e., a connected component
of $\T^\reg$.
In this paper we focus on the reduction of the dense open submanifold
of $\cM$ given by
\be
\cM^\reg = \Xi_R^{-1}(G^\reg),
\label{N16}\ee
that is,  we shall assume that in $K= b_L g_R^{-1}$ we have $g_R \in G^\reg$.
We denote
\be
\cM_\red^\reg = \{ (K,S) \mid \Lambda(K,S)=e, \,\, g_R \in G^\reg\}/G.
\label{N17}\ee

Now we state one of the  main results of the paper.

\begin{theorem}
\label{theor:4.3}
The open dense subset $\cM_\red^\reg$ of the reduced phase space can be identified
with
\be
T^*  \T^o \times \cO_B^\red,
\label{N18}\ee
where $\T^o$ is an open Weyl alcove in $\T$ and $\cO_B^\red$
is the reduced dressing orbit (\ref{N15}).  The reduced symplectic structure reads
\be
\Omega_\red = \Omega_{T^* \T^o} + \Omega_{\cO_B}^\red,
\label{N19}\ee
where the first term is the canonical symplectic form of the cotangent bundle $T^* \T^o$, and
the second term refers to the reduced orbit (\ref{N15}).
\end{theorem}

\begin{proof}  We wish to parametrize the $G$-orbits in the regular part of the constraint surface:
\be
\Lambda^{-1}(e) \cap \Xi_R^{-1}(G^\reg).
\label{N20}\ee
On account of (\ref{N7}), the action of $\eta \in G$  works on $K$ according to
\be
K= b_L g_R^{-1} \mapsto \Lambda_L(\eta b_L) \bigl(\Xi_R(\eta b_L)^{-1} g_R \Xi_R(\eta b_L)\bigr)^{-1}.
\label{N21}\ee
Since for any $b_L\in B$, the map $\eta \mapsto \Xi_R(\eta b_L)$ is a diffeomorphism on $G$,
we can transform $g_R$ into the maximal torus.  More precisely, since we assumed regularity,
we see that every gauge orbit has representatives in the set
\be
\cQ:= \{ (K, S) \mid \Lambda(K,S)=e,\,\, \Xi_R(K)  \in \T^o \}.
\label{N22}\ee
In other words, the manifold $\cQ$ is the gauge slice of a partial gauge fixing.
Now we employ the decomposition
\be
b_R = b_0 b_+ \quad \hbox{with}\quad b_0\in B_0,\, b_+ \in B_+,
\label{N*}\ee
and introduce the notation $Q:= \Xi_R(K)$.
Then the equality $K = b_L g_R^{-1} =  g_L b_R^{-1}$ tells us that
\be
g_L = g_R^{-1} = Q^{-1} \quad \hbox{and}\quad b_L = Q^{-1} b_R^{-1} Q.
\label{N**}\ee
Because of the last relation, we can write
\be
b_L b_R  = Q^{-1} b_R^{-1} Q b_R = Q^{-1} b_+^{-1} b_0^{-1} Q b_0 b_+ = Q^{-1} b_+^{-1} Q  b_+.
\label{N***}\ee
Thus, the restriction of the moment map to $\cQ$ can be expressed as
  \be
 \Lambda(K,S)= b_L b_R S= Q^{-1} b_+^{-1} Q b_+ S.
\label{N24}\ee
This has the following crucial consequences.
First, the $B_0$-factor $b_0$ of $b_R$ is not constrained.
Second, we must have $S\in B_+$, i.e.,
$S=S_+ \in \cO_B \cap B_+$.  Third, the moment map constraint
\be
Q^{-1} b_+^{-1} Q b_+ S_+ = e
\label{N25}\ee
determines $b_+$ as a  function of $Q$ and $S_+$.
To summarize, we obtain a diffeomorphism
\be
\cQ \simeq (\T^o \times B_0) \times (\cO_B\cap B_+) = \{ (Q, b_0, S_+)\}
\label{N26}\ee
by the parametrization
\be
K = Q^{-1} b_+^{-1} b_0^{-1}, \quad S = S_+ \quad \hbox{with}\quad b_+ = b_+(Q, S_+)
\label{N27}\ee
determined by the constraint equation (\ref{N25}).
We stress that, for any given
$Q \in \T^o$ and $S_+ \in \cO_B \cap B_+$,
 equation (\ref{N25}) admits a unique solution for $b_+$. (See also Section 5.)

Two elements of $\cQ$ are gauge equivalent if they are carried into each other
by the action of some $\eta \in G$. It follows from the transformation rule of $g_R$,
\be
g_R \mapsto \Xi_R(\eta b_L)^{-1} g_R \Xi_R(\eta b_L),
\label{N28}\ee
that the `residual gauge transformations' that map elements of $\cQ$ to $\cQ$
are given by the action of the subgroup $\T < G$.
The factors $Q$ and $b_0$ are invariant under this action, while $S_+$
and the corresponding $b_+$  transform according to
\be
S_+ \mapsto T S_+ T^{-1},\qquad b_+ \mapsto T b_+ T^{-1}, \quad \forall T\in \T.
\label{N29}\ee
Therefore,  recalling that $Q$ and $b_0$ can be arbitrary, we obtain the
identification
\be
\cM_\red^\reg \equiv  \cQ/\T \equiv (\T^o \times B_0) \times \cO_B^\red.
\label{N30}\ee
By general principles,
the reduced (stratified) symplectic structure on $\cM_\red^\reg$ arises
from the pull-back of the symplectic form of $\cM$ to the submanifold $\cQ$
of $\Lambda^{-1}(e)$.
Let $\iota_\cQ: \cQ \to \cM$ and $\iota_\cO: (\cO_B \cap B_+) \to \cO_B$ denote
the tautological injections, and introduce the parametrizations
\be
Q = \exp(\ri q), \quad b_0 = \exp(p),
\label{N31}\ee
where $p$ varies freely in $\cB_0$.  By using these, we find from (\ref{N2}) and (\ref{H9}) that
\be
\iota_\cQ^* (\Omega) = \langle dp \stackrel{\wedge}{,} dq \rangle + \iota_\cO^*(\Omega_{\cO_B}).
\label{N33}\ee
The second term descends to the (stratified) symplectic structure of the reduced dressing orbit (\ref{N15}).
Together with the identification (\ref{N30}), this completes the proof.
\end{proof}

\begin{remark}
\label{rem:4.4}
The ring of smooth functions on $\cO_B^\red$ (\ref{N15}) can be identified with the invariants
$C^\infty(\cO_B \cap B_+)^\T$.
Such invariants can be constructed as follows.
Let us write $S_+ \in \cO_B\cap B_+$ in
the form
\be
S_+ = \exp\bigl(\sum_{\alpha \in \Phi_+} \sigma_\alpha e_\alpha  \bigr),
\label{N34}\ee
where the $\sigma_\alpha$ are complex coordinate functions.
Consider arbitrary positive roots $\varphi_1,\dots, \varphi_{n_1}$ and $\psi_1,\dots, \psi_{n_2}$
for which
\be
\sum_{i=1}^{n_1} \varphi_i = \sum_{j=1}^{n_2} \psi_j.
\ee
Then the following polynomial function is $\T$-invariant:
\be
\prod_{i=1}^{n_1} \sigma_{\varphi_i} \prod_{j=1}^{n_2} \sigma_{\psi_j}^*.
\label{polinv}\ee
Here any repetition of the roots is allowed.
The real and imaginary parts of these complex polynomials
can be regarded as $\T$-invariant functions on the whole of $\cO_B$, by declaring that
they do not depend on $S_0$ for $S= S_0 S_+\in \cO_B$.
If we evaluate their Poisson brackets according to (\ref{B+}) and restrict
the result to $\cO_B \cap B_+$, then we obtain invariant polynomials
in the same variables $\sigma_\alpha$.
In principle, this algorithm leads to the Poisson algebra of smooth functions carried by the
reduced dressing orbit.  The reduced Poisson bracket closes on the polynomials
given by linear combinations of the invariants of the form (\ref{polinv}).
\end{remark}

\begin{remark}
\label{rem:4.5}
For completeness, it may be worth explaining
that the reduced coadjoint orbits $\cO_\red=\cO_0/\T$
  and dressing orbits $\cO_B^\red$ (\ref{N15}) are always non-empty.
For a coadjoint orbit  $\cO \subset \cG^* = \cG$,
 let us first note that $\cO \cap \cT$
is an orbit of the Weyl group of the pair $(G,\T)$. Referring to the famous convexity theorems
of Kostant, Atiyah and Guillemin and Sternberg,  one  knows that the image of the moment map for the $\T$-action
on $\cO$ is the convex hull of this Weyl orbit.
Now, let $x_i\in \cT$, $i=1,\dots, N$, denote the elements of the Weyl orbit, and form
the convex combination $x:= \frac{1}{N} \sum_{i=1}^N x_i$.
It is clear that $x$ is a fixed point for the action of the Weyl group.
But the origin is the unique fixed point, since a fixed point is characterized by the property
 that it is perpendicular to
all the roots that define the Weyl reflections.  Thus $x=0$ is in the image
of the moment map, i.e., $\cO_0$ is non-empty.

Essentially the same argument
can be applied in the case of the dressing orbits, too.
\end{remark}

\begin{remark}
\label{rem:4.6}
Recall that $\cG^*$ equipped with the linear Lie-Poisson bracket and $B$ equipped
with the multiplicative Poisson bracket (\ref{H6}) are Poisson diffeomorphic \cite{GW}.
The existence of a $\T$\emph{-equivariant} Poisson diffeomorphism
implies that every reduced dressing orbit (\ref{N15})
is symplectomorphic to a reduced coadjoint orbit $\cO_0/\T$.
Such a Ginzburg-Weinstein diffeomorphism has been exhibited in \cite{AMe} for $G=\SU(n)$.
If a $\T$-equivariant Ginzburg-Weinstein diffeomorphism exists in general,  which is  believed to be the case,
then the phase space $M_\red^\reg$ in (\ref{27}) is always symplectomorphic to the corresponding  phase space
$\cM_\red^\reg$ in (\ref{N18}).
\end{remark}

\section{Connection with the spin Sutherland model}

We need to recall some group theoretic facts.
Let $\rho: G^\C  \to  \GL(V)$ be a finite dimensional irreducible representation.
Then the complex vector space $V$ can be equipped with a Hermitian inner product in such a way
that $\rho(K^\dagger) = \rho(K)^\dagger$  holds $\forall K\in G^\C$, that is,
 the compact subgroup $G$ and $\fP$ (\ref{H15})
are represented by unitary and by positive 
operators, respectively.
The  character $\chi_\rho(K)= \tr(\rho(K))$ restricts to a $G$-invariant function on $\fP$,
and $\C^\infty(\fP)^G$ is functionally generated by the characters of the $r$ fundamental highest weight
representations.

We shall inspect the so-called {\emph{main reduced  Hamiltonians}}, which descend from the characters.
More precisely, we reduce
the $G$-invariant functions $H^\rho\in C^\infty(\cM)^G$ of the form
\be
H^\rho (K,S):= \tr_\rho(b_R b_R^\dagger):= c_\rho \tr (\rho(b_R b_R^\dagger)).
\label{F1}\ee
Here,  $K=g_L b_R^{-1}$ as in (\ref{H8}) and $c_\rho$ is a normalization constants, chosen so that
\be
c_\rho \tr\left( \rho(E_\alpha)\rho(E_{-\alpha})\right) = 2/\vert\alpha\vert^2.
\label{norm}\ee
The associated representation of $\cG^\C$ is also denoted by $\rho$,
and below we shall write simply $\tr_\rho(XY)$ instead of $c_\rho\tr( \rho(X) \rho(Y))$.

We shall demonstrate that, upon evaluation in the diagonal gauge $\cZ$ (\ref{N26}),  $H_\red^\rho$ can
be expanded in such a manner that its leading term
has the same form as  the spin Sutherland Hamiltonian (\ref{1}).
Then we shall point out the relationship between the Lax matrix
engendered by  $b_R b_R^\dagger$
and the Lax matrix of the spin Sutherland model.
In Remark 5.1, we elucidate
the interpretation of these statements in terms of a one-parameter deformation.
In Section 5.2,
we derive explicit formulas for $G^\C = \SL(n,\C)$, using its defining
representation.

\subsection{Reduced main Hamiltonians and Lax matrices}

Let us inspect the constraint equation (\ref{N25})
by parametrizing the variables as
\be
S_+ = e^\sigma, \quad b_+ = e^\beta, \quad
\sigma = \sum_{\alpha>0} \sigma_\alpha E_\alpha,\quad \beta = \sum_{\alpha>0} \beta_\alpha E_\alpha
\label{F2}\ee
using complex expansion coefficients $\sigma_\alpha$, $\beta_\alpha$, and  $Q =e^{\ri q}$.
The Baker-Campbell-Hausdorff formula permits us to rewrite the constraint equation as
\be
\exp( \beta - Q^{-1} \beta Q - \frac{1}{2} [ Q^{-1} \beta Q, \beta] + \cdots ) = \exp(-\sigma),
\label{F3}\ee
where the dots indicate higher commutators. Note that the BCH series is now finite, since $\cB_+$ is
nilpotent. Using that $B_+$ is diffeomorphic to its Lie algebra by the exponential map,
we see from
\be
 \beta - Q^{-1} \beta Q - \frac{1}{2} [ Q^{-1} \beta Q, \beta] + \cdots = - \sigma
\label{F4}\ee
that $\beta_\alpha$ can be expressed in terms of $\sigma$ and $e^{\ri q}$ in the
following form:
\be
\beta_\alpha = \frac{\sigma_\alpha}{e^{-\ri \alpha(q)} -1 } + \Gamma_\alpha(e^{\ri q}, \sigma),
\label{F5}\ee
where $\Gamma_\alpha$ contains higher order terms in the components of $\sigma$. Namely,
we have
\be
\Gamma_\alpha = \sum_{k\geq 2} \sum_{\varphi_1,\dots, \varphi_k} f_{\varphi_1,\dots, \varphi_k}(e^{\ri q})
\sigma_{\varphi_1}\dots \sigma_{\varphi_k},
\label{F6}\ee
where the $\varphi_1,\dots, \varphi_k$ are unordered collections of
positive roots, such that $\alpha = \varphi_1 + \dots + \varphi_k$.
The $f_{\varphi_1,\dots, \varphi_k}$ are rational functions in $e^{\ri q}$, and the sum contains
only finitely many terms.
These claims are established by inspection of equation (\ref{F4}), proceeding
iteratively according to the height of the root $\alpha$ labeling $\sigma_\alpha$ on the right-hand side.
It follows that
the restriction of $H^\rho$ (\ref{F1})  to $\cQ \subset \Lambda^{-1}(e)$ (\ref{N26}) can be expanded as a
polynomial in the components of the `spin variable' $\sigma$, with coefficients given by
rational functions in the components of $e^{\ri q}$ multiplied by
suitable components of $b_0^2 \equiv \exp(2 p)$.
Specifically, we see that
\be
H^\rho_\red(e^{\ri q},p,\sigma) = \tr_\rho( e^{2 p} b_+ b_+^\dagger)
\label{F7}\ee
 can be expanded as
\be
H_\red^\rho(e^{\ri q},p,\sigma)= \tr_\rho\left(e^{2p}\left(\1_\rho + \frac{1}{4}\sum_{\alpha>0}
\frac{ \vert \sigma_\alpha \vert^2 E_{\alpha} E_{-\alpha}}{\sin^2(\alpha(q)/2)} +
\o_2(\sigma,\sigma^*) \right)\right).
\label{F8}\ee
To obtain this from (\ref{F7}), we took into account that $E_\alpha^\dagger = E_{-\alpha}$ and that,
in consequence of the invariance of $\tr_\rho$ with respect to conjugation by the elements  of $\T$,
 $\tr_\rho(e^{2p} E_\alpha E_{-\gamma}) =0$ unless $\gamma= \alpha$.
Applying (\ref{norm})  and expanding $e^{2p}$,  equation (\ref{F8}) implies
\be
H^\rho_\red(e^{\ri q},p,\sigma) = \dim_\rho + 2 \tr_\rho (p^2) +
\frac{1}{2} \sum_{\alpha>0} \frac{1}{\vert \alpha\vert^2} \frac{ \vert \sigma_\alpha \vert^2 }{\sin^2(\alpha(q)/2)}  +
 \o_2(\sigma,\sigma^*, p).
\label{F9}\ee
The symbols $\o_2$ indicate terms that are at least cubic in the components of
the displayed arguments, and depend rationally on $e^{\ri q}$.
We used that $\tr_\rho(p) = 0$, because $\cG^\C$ is a simple Lie algebra.
According to (\ref{F9}), the
 leading term of $\frac{1}{4} (H_\red^\rho - \dim_\rho)$
matches precisely the spin Sutherland Hamiltonian (\ref{tilde1}).

We have seen that, when expanded in the components $p$ and the spin variable $\sigma$,
the leading term of the reduced main Hamiltonian (\ref{F9}) has the same form as the spin Sutherland
 Hamiltonian (\ref{tilde1}).
It should be noted that the Poisson brackets of the corresponding spin variables, $\tilde{\xi}_\alpha$ and $\sigma_\alpha$,
are different in the two cases.
The relationship between the spin algebra that arises in our Poisson-Lie case
and the one that occurs for the spin Sutherland model is given by equation (\ref{B+}).
It follows that we recover the spin Sutherland model (\ref{1}) from our reduced system
if we keep only the leading term of the reduced  Hamiltonian $H^\rho_\red$ and only the leading terms
of the Poisson brackets of the spin variables.

All reduced Hamiltonians arising from $C_I(\cM)= \Lambda_R^*(C^\infty(B)^G)$ can be  expressed as $G$-invariant
functions of the `Lax matrix' $L(e^{\ri q}, p, \sigma):= b_R b_R^\dagger$,
regarded as a function on $\cZ$ (\ref{N26}).
Working in some representation and keeping only linear terms in $p$ and $\sigma$,
 this Lax matrix can be expanded as follows:
\be
L(e^{\ri q}, p, \sigma )= \1 + 2p + \sum_{\alpha>0} \left(\frac{\sigma_\alpha}{e^{-\ri \alpha(q)} -1} E_{\alpha} +
\frac{\sigma_\alpha^*}{e^{\ri \alpha(q)} -1} E_{-\alpha}\right)
+ \o(\sigma, \sigma^*, p).
\label{LqpS}\ee
If we replace $\sigma_\alpha$ by $\tilde \xi_\alpha$ and substitute
the change of variables (\ref{tildevar}), then (up to irrelevant constants) the linear part of $L$ in $\sigma, \sigma^*$ and $p$
reproduces $L_{\mathrm{Suth}}:=  \ri J$ as given by equations (\ref{214}) and (\ref{215}), which can be regarded as a Lax matrix
for the spin Sutherland model:
\be
L_{\mathrm{Suth}}(e^{\ri q},p, \xi) =  p
+ \sum_{\alpha>0} \left(\frac{\ri \xi_\alpha}{e^{-\ri \alpha(q)} -1} E_{\alpha} +
\frac{(\ri \xi_\alpha)^*}{e^{\ri \alpha(q)} -1} E_{-\alpha}\right).
\ee
It is clear from this that the $G$-invariant polynomials of $L_{\mathrm{Suth}}$
can be recovered as leading terms of suitable invariant polynomials of $(L - \1)$.

In fact, taking $G=\SU(n)$ and a minimal dressing orbit our construction
reproduces
the standard trigonometric RS system \cite{FK1}. In this case,
there are no spin variables since the reduced dressing orbit consists of a single point.
The characters of the fundamental representations give rise  to $(n-1)$ independent Hamiltonians,
which become the standard RS Hamiltonians \cite{RBanff,RS} after a certain canonical transformation.
In the general case, the Hamiltonians (\ref{F8}) resemble the
light-cone Hamiltonians of the RS system. The presence of the leading exponent $e^{2p}$ in (\ref{F8})
entails that the reduced main Hamiltonians originating from
the $r$ fundamental characters {\emph{are  functionally independent in general}}.

\begin{remark}
\label{rem:5.1}
The foregoing results justify calling our models generalizations
of the spin Sutherland models, but we also would like to view them as \emph{deformations}.
However, so far we have not introduced any deformation parameter.
In fact, we can not introduce a deformation parameter if we take the spin
variables from a fixed reduced dressing orbit.
Now we explain that  this obstacle can be overcome by
placing all our reduced phase spaces (\ref{N18}) inside  the single Poisson space
$T^* \T^o \times (B_+/\T)$,
where $B_+/\T$ is the Poisson reduction of $(B, \{\ ,\ \}_B)$  at the zero value
of the moment map for the Hamiltonian $\T$-action. We can then consider a real, positive scale parameter
$\epsilon$ and replace the variables
$(e^{\ri q}, p, e^{\sigma}) \in T^* \T^o \times B_+$ by $(e^{\ri q}, \epsilon p, e^{\epsilon \sigma})$.
At the same time, we use the scaled Poisson structure $\{\ ,\ \}^\epsilon = \epsilon \{\ ,\ \}$,
so that $p$ and $q$ remain canonical conjugates with respect to $\{\ ,\ \}^\epsilon$.
In this way we obtain the relation
\be
\lim_{\epsilon \to 0} \frac{1}{4\epsilon^2}\left( H_\red^\rho(e^{\ri q}, \epsilon p, \epsilon \sigma) - \dim_\rho\right)=
\frac{1}{2} \tr_\rho (p^2) +
\frac{1}{8} \sum_{\alpha>0} \frac{1}{\vert \alpha\vert^2} \frac{ \vert \sigma_\alpha \vert^2 }{\sin^2(\alpha(q)/2)}.
 \ee
Moreover, it follows from (\ref{B+}) that in the $\epsilon=0$ limit the $\T$-invariant polynomials of
 $\sigma\in \cB_+$ satisfy the Poisson brackets corresponding to the reduction of $\cG^* \simeq \cB$
at the zero value of the moment map for the respective $\T$-action.
The expansion  (\ref{LqpS}) of $L$ gives
 \be
 L(e^{\ri q}, \epsilon p, \epsilon \sigma) = \1 + 2\epsilon  L_{\mathrm{Suth}}(e^{\ri q}, p, \xi)+ \o(\epsilon),
 \quad
 \hbox{with}\quad
 2\ri\xi_\alpha = \sigma_\alpha,
 \ee
 and thus we can recover all invariant polynomials of $L_{\mathrm{Suth}}$ via suitable
 scaling limits of invariants of  $L(e^{\ri q}, \epsilon p, \epsilon \sigma)$.
For example, if we consider $G^\C=\SL(n,\C)$ then we obtain $(n-1)$ independent
Hamiltonians for the spin Sutherland model from
\be
\tr(L_{\mathrm{Suth}}(e^{\ri q},p, \xi)^k) = \lim_{\epsilon \to 0} \frac{1}{(2\epsilon)^k}
\tr \left( (L(e^{\ri q}, \epsilon p, \epsilon \sigma) - \1_n)^k\right),\qquad
k=2,\dots, n.
\ee
The limiting procedure just outlined is similar to the way in which the
spinless Sutherland model is obtained as the non-relativistic limit
 of the standard RS model \cite{RBanff,RS}.
In that context,  the parameter $\epsilon^{-1}$ plays the role of the
`velocity of light'.
\end{remark}

\begin{remark}
\label{rem:5.2}
One-parameter families of dressing orbits of a given type are represented by the $G$-orbits in $\fP$ that
are of the form  $\exp(2\ri s \cO_\cG)$ for a fixed coadjoint orbit  in $\cG^*\equiv \cG$ and any
non-zero real parameter $s$.  It is tempting to speculate that $s$ can also give rise to a
deformation parameter, which should connect the spin RS  models associated with these dressing orbits
to the spin Sutherland model based on $\cO_\cG$.
This appears an intricate issue because of
the complicated mapping from the exponential
parametrization $S S^\dagger = \exp(2 \ri s X)$ with  $X\in \cO_\cG$ (cf. equation (\ref{H19}))
to the parametrization $S=  S_0 e^\sigma$ that we found convenient to use.
We plan to return to this problem elsewhere.
\end{remark}

\subsection{Explicit formulas for $G^\C= \SL(n,\C)$.}

In this case, it is convenient to parametrize $b\in B$ by using matrix elements.
After diagonalizing  $g_R$, we are going to solve the moment map constraint for
\be
b_R = e^{p} b,
\label{bRpar}\ee
where $b$ is an unknown upper triangular matrix with unit diagonal.
The constraint equation is given by
\be
Q^{-1} b Q = b S,
\ee
where $Q= \diag(Q_1,\dots, Q_n)\in \T^o$ and $S\in \cO_B\cap B_+$ is the constrained `spin' variable.
Except for the notation, this is the same as equation (\ref{N25}).
Applying the principal gradation of $n\times n$ matrices, which is additive under
matrix multiplication,  we expand $b$ and $S$ according to
\be
b= \1_n + b(1) + \cdots + b(n-1), \quad
S= \1_n + S(1) + \cdots + S(n-1).
\ee
For example, $S(j)$ is a linear combination of the $n\times n$ elementary matrices of the form
$E_{i, i+j}$.
We can spell out the constraint equation as
\be
Q^{-1} b(k) Q - b(k) = S(k) + \sum_{i+j = k} b(i) S(j),
\qquad k=1,\dots, n-1,
\ee
where $1\leq i,j \leq n-1$.
This can be solved iteratively, proceeding upwards from $k=1$ until $k=n-1$.
To write down the solution, we introduce the shorthand notation
\be
\cI_{a, a+ j} = \frac{1}{Q_{a+j} Q_a^{-1} - 1},
\ee
which is well defined since $Q$ is regular.
It is not difficult to ascertain that the solution for the matrix $b$ is as follows.
For the grade $1$ matrix entries, we have
\be
b_{a,a+1} = \cI_{a, a+1} S_{a,a+1}.
\ee
For higher grades, the matrix elements $b_{a, a+k}$,  $k=2,\dots, n-a$, read as follows:
\be
b_{a, a+k} = \cI_{a, a+k} S_{a, a+k} + \sum_{\substack{m=2,\dots, k\\ (i_1,\dots, i_m)\in \N^m\\ i_1 + \dots + i_m = k} }
\prod_{\alpha=1}^m
\cI_{a, a + i_1 +\dots + i_\alpha} S_{a + i_1 + \dots + i_{\alpha -1}, a+ i_1 + \dots + i_\alpha}.
\label{explicit}\ee
It is understood that $i_0=0$ and the $S$-factor for $\alpha =1$ is $S_{a, a+ i_1}$.
We omit the inductive verification of this formula, which can be done using
that the form of $b_{a,a+k}$,  for $a+k \leq n$, does not depend on $n$.

By substituting (\ref{bRpar}) and (\ref{explicit}) into $\tr (b_R b_R^\dagger)$ we obtain the expansion
\be
H_\red(e^{\ri q},p,S) = \tr(b_R b_R^\dagger) =\sum_{a=1}^n  e^{2p_a} +
\sum_{a=1}^{n-1} e^{2p_a} \sum_{k=1}^{n-a} \vert \cI_{a,a+k} S_{a, a + k}\vert^2 + \hbox{higher order terms},
\label{HredU}\ee
where the higher order terms are at-least-cubic polynomials in the matrix elements of $S$ and their complex conjugates.
Since $Q_a = e^{\ri q_a}$, we have
\be
\vert  \cI_{a, a+ k}\vert^2  =\frac{1}{4 \sin^2((q_{a+k} - q_a)/2)},
\ee
and therefore
\be
H_\red (e^{\ri q},p,S) =\sum_{a=1}^n  e^{2p_a} +
\frac{1}{4}\sum_{a=1}^{n-1} e^{2p_a} \sum_{k=1}^{n-a} \frac{ \vert S_{a, a + k}\vert^2}{ \sin^2((q_{a+k} - q_a)/2)}
 + \hbox{higher order terms}.
\ee
This generalizes the corresponding spin Sutherland Hamiltonian, since
\be
H_\red (e^{\ri q},p,S) -n  =2 \sum_{a=1}^n  p_a^2 +
\frac{1}{4} \sum_{a=1}^{n-1}  \sum_{k=1}^{n-a} \frac{ \vert S_{a, a + k}\vert^2}{ \sin^2((q_{a+k} - q_a)/2)}
 + \cdots
\ee
with the dots standing for terms whose total degree in $p$, $S$ and $S^*$ is at least $3$.

Since by (\ref{explicit}) we have solved the constraints, we could write a fully explicit formula for
the reduced main Hamiltonian (\ref{HredU}), represented as a $\T$-invariant function on $\cZ$ (\ref{N26}).
We present it as an illustration in the simplest $n=3$ case, for which we have
\be
b_{1,2} = \cI_{1,2} S_{1,2},\quad b_{2,3}=\cI_{2,3} S_{2,3},\quad
b_{1,3}= \cI_{1,3} \left(S_{1,3} + \cI_{1,2}  S_{1,2} S_{2,3}\right)
\ee
and, with $p_1 + p_2 + p_3 =0$,
\bea
&&H_{\red}= (e^{2p_1} + e^{2p_2} + e^{2p_3}) + e^{2p_1} \left(\vert\cI_{1,2} S_{1,2} \vert^2 + \vert \cI_{1,3} S_{1,3} \vert^2 \right)
+ e^{2p_2} \vert\cI_{2,3} S_{2,3} \vert^2   \nonumber\\
&&\qquad \qquad + 2 e^{2p_1} \vert\cI_{1,3}\vert^2 \operatorname{Re}\left(\cI_{1,2}  S_{1,2} S_{2,3} S_{1,3}^* \right).
\eea

In principle,  the structure of the reduced Poisson algebra can be made explicit too,
using that the Poisson structure of $B$ can be described in terms of the matrix elements of
the freely varying $b \in B$.  As one can check, and is well-known, this structure is given by
\be
\{ b_{m,j}, b_{k,l}\}_B =\ri b_{k,j} b_{m,l} \left[ \delta_{m,k} + 2 \theta(m-k) - \delta_{l,j} - 2 \theta(l-j)\right],
\ee
where $\theta$ is the step function (+1 for positive arguments, and 0 for non-positive ones), and
\be
\{ b_{m,j}, b_{k,l}^* \}_B = \ri b_{m,j} b_{k,l}^* [ \delta_{m,k} - \delta_{j,l}]
+ \ri \left[ \delta_{m,k} \sum_{\beta > m} b_{\beta, j} b_{\beta, l}^* - \delta_{j,l} \sum_{\alpha <j} b_{m,\alpha}
b_{k,\alpha}^*\right].
\ee
This fixes the real Poisson brackets of the real and imaginary parts of the matrix elements,
and allows us to find the Poisson algebra of their $\T$-invariant polynomials.

\section{Reduced Hamiltonian flows and integrability}

Here, we  derive the equations of motion of the reduced system
and present an algorithm for obtaining its solutions. We then point out the connection
between our construction and  the work of L.-C.~Li  \cite{Li1,Li2}.
Finally, we briefly discuss the integrability
of the system.

\subsection{Equations of motion and solutions by projection}

For  any $G$-invariant function $H$ on $\cM$, the reduced Hamiltonian flow can be
obtained by first restricting the original flow to the moment map constraint surface, and then projecting it to $\cM_\red$.
The reduced Hamiltonian vector field can be obtained in an analogous manner.
If the reduced phase space is realized via gauge fixing, this means that in general we have to add to the Hamiltonian
vector field of $H$ a point dependent infinitesimal
gauge transformation\footnote{This `compensation' is needed if
the restricted Hamiltonian vector field is not tangent to the gauge slice.}
 that ensures that the additional
constraints are preserved.
We can apply the same procedure to partial gauge fixings as well, like the `diagonal gauge' $\cZ$ (\ref{N22}),
where $g_R$ belongs to $\T^o$.

The determination of the `compensating gauge transformation' relies on the following.
Let $X_\cM$ be the vector field on $\cM$ that gives the infinitesimal action of $X\in \cG$ generated
by the moment map $\Lambda$. Regarding $b_R$, $g_R$, $S$ etc as evaluation functions on $\cM$, their Lie-derivatives
with respect to this vector field can be easily evaluated.
We obtain
\be
\cL_{X_\cM}(b_R) = \dress_{(b_L^{-1} X b_L)_\cG} (b_R),\quad
\cL_{X_\cM}(g_R) = [(b_L^{-1} X b_L)_\cG, g_R].
\label{der1}\ee
The first formula and (\ref{dress})  imply
\be
\cL_{X_\cM}(b_R b_R^\dagger) = [(b_L^{-1} X b_L)_\cG, b_R b_R^\dagger],
\label{der2}\ee
and for completeness we also record
\be
\cL_{X_\cM}(S) = \dress_{(b_R^{-1} b_L^{-1} X b_L b_R)_\cG} (S),\quad
\cL_{X_\cM}(b_R S) = \dress_{(b_L^{-1} X b_L)_\cG}(b_R S).
\label{der3}\ee
Note that $X\mapsto (b_L^{-1} X b_L)_\cG$ is an invertible map on $\cG$. Therefore, any
 point dependent infinitesimal gauge transformation can be represented as a vector field of
 the above form, where $(b_L^{-1} X b_L)_\cG$ is replaced by some $\cG$-valued function,
 called $Y$ below, on the phase space.

Now, consider a function $H\in C^\infty(\cM)^G$  of the form $H(K,S) = h(b_R)$, and
denote
\be
 (d^L h)(b_R) =: \cV(L)\qquad \hbox{with}\quad L:= b_R b_R^\dagger.
\label{cV}\ee
Before reduction, the  Hamiltonian vector field of $H$ can be symbolically
written as
\be
\dot{g}_R = \cV(L) g_R, \quad
\dot{b}_R=0,\quad \dot{S}=0.
\label{unred}\ee
For the description of the reduced symplectic structure, it was appropriate to use
the variables $Q= e^{\ri q}$, $p$ and $S$, but the equations of motion are more
conveniently expressed in terms of $Q$ and $L$.
Since we work in the diagonal gauge, $p$ can be uniquely recovered from
$b_R$, which itself is uniquely determined by $L$, via either of the unique decompositions
\be
b_R = e^{p} b_+ = n_+ e^p \quad \hbox{with}\quad b_+, n_+\in B_+.
\label{bRdec}\ee
Taking advantage of the moment map constraint (\ref{N25}), $S=S_+$ is determined by the formula
\be
S= b_R^{-1} Q^{-1} b_R Q.
\label{Sfun}\ee
Thus,  it is enough to find the time development
of $Q$ and $L$, and then recover $p$ and $S$ by the above algebraic relations.

Corresponding to $\cG = \cT + \cT^\perp$, we decompose the $\cG$-valued functions
$\cV$ and $Y$ as
\be
\cV = \cV_\cT + \cV_\perp
\quad\hbox{and} \quad Y = Y_\cT + Y_\perp.
\ee
Then we can state

\begin{proposition}
\label{prop:6.1}
The projection of the Hamiltonian vector field (\ref{unred}) to the diagonal gauge
$\cZ$ (\ref{N22}) gives
\be
\dot{Q}  = \cV_\cT(L) Q
\qquad\hbox{and}\qquad \dot{L}= [Y_\cT + Y_\perp(Q,L), L],
\label{eq1}\ee
where
\be
Y_\perp(Q,L) = \left(\Ad_Q - \id\right)^{-1}(\cV_\perp(L))
\label{eq2}\ee
and $Y_\cT$ is arbitrary, reflecting the residual gauge transformations.
\end{proposition}

\begin{proof}
By restriction to $\cZ$, where $g_R=Q$, and addition of a vector field tangent to the gauge orbits,
we obtain from (\ref{unred}) and (\ref{der1}), (\ref{der2}) the `corrected' derivatives
\be
\dot{Q}  = \cV(L) Q + [ Y, Q], \qquad \dot{L} = [Y, L].
\label{dotQL}\ee
The tangency to $\cZ$ means that $\dot{Q} Q^{-1}$ must belong to $\cT$, which
is equivalent to
\be
Q Y_\perp Q^{-1} - Y_\perp  = \cV_\perp(L).
\ee
Since $Q\in \T^o$ is regular, $(\Ad_Q-\id)$ is invertible on $\cT^\perp$, and
therefore $Y_\perp$ is uniquely found to be given by (\ref{eq2}).
The  constituent $Y_\cT$ of $Y$ is left undetermined,
since we have not fixed the residual gauge transformations that map $\cZ$ to $\cZ$.
\end{proof}

In principle, one can construct the solutions of the equations of motion
by projecting the unreduced flow (\ref{H13})  as follows.
Pick any initial values $Q(0), L(0)$ defined by a point of $\cZ$.
Then, there exists $\eta(t) \in G$ and unique $Q(t)$, in general from the
closure of the Weyl alcove $\T^o$,  verifying
\be
Q(t) = \eta(t) \exp(t \cV(L(0))) Q(0) \eta(t)^{-1}.
\ee
If $ \exp(t \cV(L(0))) Q(0)$ belongs to $G^\reg$,
which
certainly holds for small $t$, then one can choose $\eta(t)$ in such a way to depend analytically on $t$ and on the
initial values.
The resulting $Q(t)\in \T^o$ and $L(t) = \eta(t) L(0) \eta(t)^{-1}$ solve
the equations of motion (\ref{eq1}).
Along this solution, one has
\be
Y(t) = \dot{\eta}(t) \eta(t)^{-1}.
\ee
We note that $\cV(L(t)) = \eta(t) \cV(L(0)) \eta(t)^{-1}$ follows, since $h\in C^\infty(B)^G$ (see \ref{cV})).
As $\eta(t)$ can be replaced by
$T(t)\eta(t)$ with an arbitrary function $T(t) \in \T$,
one can eliminate $Y_\cT$ for any given solution.

According to (\ref{bRdec}), the time development of $p$ can be recovered from the decomposition
\be
L(t) = b_R(t) b_R(t)^\dagger = n_+(t) e^{2 p(t)} n_+(t)^\dagger, \quad \hbox{with}\quad n_+(t) \in B_+.
\ee
After finding $Q(t)$, and $b_R(t)$ from $L(t)$,
the time development of $S$ follows from (\ref{Sfun}).

It requires further delicate analysis to see whether
and for which dressing orbits can
 the trajectory $g_R(t) = \exp(t \cV(L(0)) Q(0)$ leave the set of regular elements.
 Then  some aspects of our analysis break down, although $Q$ in the closure of $\T^o$
 still remains a well-defined gauge invariant object.
 Moreover, it follows from general results \cite{SL,Sn} that
 the reduced Hamiltonian flows are all complete, and they
 stay inside the symplectic strata of the full reduced phase space.
  We plan to investigate this issue, first in examples,
 in a future work.

\begin{remark}
\label{rem:6.2}
For regular $Q=e^{\ri q}$, let $R(Q) \in \mathrm{End}{(\cG)}$ be the following linear operator:
\be
R(Q)|_{\mathcal{\cT}}  :=0, \qquad
R(Q)|_{\mathcal{\cT^\perp}}:=
\frac{1}{2} (\Ad_Q + \id) \left((\Ad_Q - \id)|_{\mathcal{\cT}^{\perp}}\right)^{-1} = \frac{1}{2}
\coth\left(\frac{1}{2} (\ad_{\ri q})|_{\mathcal{\cT^\perp}} \right).
\label{RQ}
\ee
Viewing it as a function of $\ri q\in \cT_\reg$,  this is a solution of the modified
classical dynamical Yang--Baxter equation \cite{EV}.   The equations of motion
(\ref{eq1}) admit the equivalent form
\be
\dot{Q}  = \cV_\cT(L) Q
\qquad\hbox{and}\qquad  \dot{L}= [Y_\cT' + R(Q)(\cV(L)), L].
\label{eq3}\ee
Here, $Y'_\cT = Y_\cT + \frac{1}{2} \cV_\cT(L)$, which is as arbitrary as $Y_\cT$ is, and we
 used that, in consequence of the definitions, $[\cV(L), L]=0$.
We recall \cite{FP1} that
the same dynamical $R$-matrix features in the equations of motion of the spin Sutherland models of Section 2.
However,
 for those models the Lax matrix belongs to $\ri \cG$, while our $L$ lies in $\exp(\ri \cG)$.
\end{remark}

As an illustration, let us consider
\be
h:= \tr( b_R b_R^\dagger)
\ee
for $G^\C_\R= \SL(n,\C)$. The form of the corresponding reduced Hamiltonian was presented in the preceding section.
To display the reduced equations of motion, note that
in this case
\be
\cV(L) = 2 \ri L - \frac{2\ri}{n} \tr(L) \1_n.
\ee
Writing $Q = \exp\left(\diag(\ri q_1,\dots ,\ri q_n)\right)$,  (\ref{eq2}) gives
\be
Y_{jk} = \left( \cot \frac{q_j - q_k}{2} - \ri\right) L_{jk}\quad \hbox{for}\quad 1\leq j\neq k\leq n.
\ee
Substitution into (\ref{eq1}) yields the equations of motion:
\be
\dot{q}_j = 2L_{jj} - \frac{2}{n} \tr(L), \qquad
 \dot{L}_{jj} = 2\sum_{\ell \neq j} \vert L_{j\ell} \vert^2 \cot\frac{q_j - q_\ell}{2},
\ee
and, for $1\leq j\neq k\leq n$,
\bea
&&\dot{L}_{jk} =
 \sum_{\ell \neq j} L_{j\ell}  L_{\ell k} \cot\frac{q_j - q_\ell}{2} -
 \sum_{\ell \neq k}L_{j\ell}  L_{\ell k} \cot\frac{q_\ell - q_k}{2}\nonumber \\
&&\qquad\qquad +(Y_{jj} - Y_{kk}  + \ri L_{jj} - \ri L_{kk}) L_{jk}.
\eea
Since the diagonal entries of $Y$ are arbitrary, one can choose them in such a way
to get rid of the second line in $\dot{L}_{jk}$.
After analytic continuation from trigonometric to hyperbolic functions,
these equations coincide
 with Eq.~(15) in \cite{BH} as well as with Eq. (5.15) in \cite{Li1}.
They are also special cases of Eq. (3.36) in \cite{AF} if the `potential' that appears there
is taken to be simply $\cot(q)$ instead of $\cot(q) - \cot (q+\gamma)$, i.e., if we omit the
term containing the coupling constant $\gamma$.

\begin{remark}
\label{rem:6.3}
In the papers \cite{Li1,Li2}, L.-C.~Li constructed certain  hyperbolic  spin RS type models.
The construction begins with  Hamiltonian systems on $\cT_\reg^\C \times G_\R^\C \times \cT_\reg^\C$,
which is a coboundary dynamical Poisson groupoid based on a classical dynamical $R$-matrix.
The $R$-matrix in question is a map from $\cT_\reg^\C$ to $\mathrm{End}(\cG_\R^\C)$ that can be chosen
to be a natural extension of the $R$-matrix in (\ref{RQ}).
The next step involves restriction to the fixed point set of a suitable Poisson involution,
and subsequent Hamiltonian reduction.
In the case described in most detail in \cite{Li2}, the fixed point set has the form $\ri \cT_\reg \times \fH$, where
$\fH \subset G_\R^\C$ is the fixed point set of the inverse of the Cartan involution, i.e., the set of Hermitian
elements. Upon analytic continuation from hyperbolic to trigonometric functions,
the equations of motion, worked out in \cite{Li1,Li2} for $G^\C=\SL(n,\C)$, coincide with the equations presented above.
Moreover, commuting Hamiltonians are obtained in \cite{Li2} from the $G$-invariant
functions on $\fH$, similarly to our usage of $C^\infty(\fP)^G$.
However, because the constructions are very different,   it is not clear
at present whether the reduced Poisson structures are also related by analytic continuation.
We expect that this is the case, but it requires further effort to show it.
We note that neither the distribution of the variables as $q,p$ and $S$, nor
the connection between the spin RS type models and
spin Sutherland models  appear in \cite{Li1,Li2}.
\end{remark}

\subsection{Integrals of motion and Lax equations}

We noted in Section 4 that the unreduced free system is degenerately integrable.
The unreduced free Hamiltonians form the Abelian Poisson algebra
\be
C_I(\cM) = \Lambda_R^*\left(C^\infty(B)^G\right).
\ee
The elements of $C_J(\cM)$ (\ref{N3}) represent integrals of motion,
since they Poisson commute with every element of $C_I(\cM) \subset C_J(\cM)$.
It is known that, generically, degenerate integrability is inherited under
Hamiltonian reduction based on a compact symmetry group.
In the $C^\infty$ category, it is also known that degenerately
integrable systems are integrable in the usual Liouville sense, too.
For precise statements, see the reviews \cite{J,Zung2} and references therein.
The arguments of Reshetikhin \cite{Res1,Res2},  which show degenerate integrability by proving
the existence of enough analytic integrals of motion  on a dense open subset
of the reduced phase space, can be also adapted to our case.

However, for a concrete integrable system one is not content with existence statements,
but would like to obtain the required integrals of motion in explicit form.
The elements of $C_I(\cM)$ are  the $G$-invariant functions of the `unreduced Lax matrix' $b_R b_R^\dagger$,
and they
descend to the reduced phase space, since $C_I(\cM) \subset C^\infty(\cM)^G$.
Now we exhibit a large class of $G$-invariant
elements of $C_J(\cM)$.
To this end, let us consider an arbitrary polynomial
\be
\cP(b_R b_R^\dagger, g_R^{-1} b_R b_R^\dagger g_R)
\ee
in the two non-commutative variables that belong to $\fP$ (\ref{H15}).
Take an arbitrary finite dimensional representation  $\rho$ of $G_\R^\C$, whose restriction to  $G$ is unitary.
Then the real an imaginary parts of the function $\cF$, given on $\cM$ by
\be
\cF(K,S):=
\tr_\rho\bigl(\cP(b_R b_R^\dagger, g_R^{-1} b_R b_R^\dagger g_R)\bigr),
\ee
are $G$-invariant elements of $C_J(\cM)$.
In order to see this, one may use the relation
\be
g_R^{-1} b_R b_R^\dagger g_R = b_L^{-1}  (b_L^{-1})^\dagger,
\ee
which follows from (\ref{H8}),
and the fact that $b_R b_R^\dagger $ transforms in the same way as $g_R$ (\ref{N28}).
These $G$-invariant integrals of motion
descend to smooth (even real analytic, and in a certain sense algebraic) functions
on the full reduced phase space. They are generalizations of the conserved
quantities $\tr_\rho(\cP(J, g J g^{-1}))$  displayed in Section 2.

Upon imposing the moment map constraint and working in
the diagonal gauge, where $g_R=Q \in \T^o$,
we obtain
\be
g_R^{-1} b_R b_R^\dagger g_R = Q^{-1} L Q
\quad\hbox{with}\quad L\equiv b_R b_R^\dagger.
\label{Zrel}\ee
Thus the above constants of motion take the form
$\tr_\rho\left( \cP(L, Q^{-1} L Q)\right)$.
Special cases of these conserved quantities  can be generated
by evaluation of the trace on powers of the following spectral parameter (denoted by $\lambda$) dependent
Lax matrix:
\be
\cL := L + \lambda Q^{-1} L Q.
\label{Lax12}\ee
Working in the diagonal gauge, the Hamiltonian vector field of $H\in C_I(\cM)$,
for which $H(K,S) = h(b_R)$, implies the  Lax equation
\be
\dot{\cL} = [ Y, \cL].
\label{Laxeqs12}\ee
Here, $Y$ is the $\cG$-valued function responsible for consistency with the (partial) gauge fixing.
Equation (\ref{Laxeqs12}) follows directly from (\ref{dotQL}) using that $[\cV(L), L]=0$  because
$h\in C^\infty(B)^G$.

It is an open  problem if the above exhibited integrals of motion are sufficient for the degenerate integrability
of the reduced system.
The question of Liouville integrability with real analytic integrals of motion also
requires further study.

\section{Discussion and outlook}

In this paper we applied Hamiltonian reduction to  `free systems'
on Heisenberg doubles, generalizing the derivation of the trigonometric spin Sutherland
models (\ref{1}) based on the corresponding cotangent bundles.
The reduced systems that we obtained appear
to be related by analytic continuation to
certain  hyperbolic spin RS type models introduced
earlier by L.-C.~Li in \cite{Li2}\footnote{Incidentally,
we noticed the connection with \cite{Li2}
only during the completion of this manuscript.}.
The different approaches have various advantages with respect to each other.

Our approach can, in principle,
accommodate action-angle duals of all of our systems.
Indeed, the unreduced phase space $\cM$ supports \emph{two} natural Abelian Poisson algebras,
namely $\Lambda_R^*(C^\infty(B)^G)$ and $\Xi_R^*(C^\infty(G)^G)$,
which define two degenerately integrable systems.
We worked out a model of the reduced phase space whereby the reduction
of $\Lambda_R^*(C^\infty(B)^G)$ represents many-body Hamiltonians of spin RS type, and $\Xi_R^*(C^\infty(G)^G)$
reduces to their position variables.
There should exist another model of the same reduced phase space,
making it possible to view the reduction of $\Xi_R^*(C^\infty(G)^G)$ as many-body Hamiltonians with
their position variables descending from   $\Lambda_R^*(C^\infty(B)^G)$.
This is the standard way to produce pairs of many-body systems in action-angle duality via reduction,
 see e.g.~\cite{AF,FK2,N,Res1}.
 Another potential advantage of our method is that it automatically leads to complete flows of the interesting Hamiltonians
on the full reduced phase space, simply since  the unreduced flows are complete.
This may not hold in the framework of \cite{Li1,Li2}, since it relies on dynamical
$R$-matrices, whose domain of definition requires a restriction
of the variables. In our present work, the restriction to regular elements of $G$
was merely a technical convenience, and we shall investigate the global
features of our reduced systems in the future.

Here, it is proper to stress that the method of \cite{Li1,Li2} is applicable for a large class of dynamical $R$-matrices,
not only for the standard trigonometric $R$-matrix, which is related to our case.
We plan to report the precise relationship between our systems and those in \cite{Li2} in a future publication.

In passing, we wish to mention  that other approaches are also available for constructing
trigonometric/hyperbolic spin RS type systems.  In particular, the moduli space of
flat $G^\C$ connections on a torus with a hole is believed to support such systems
at arbitrary monodromy.  These moduli spaces can be described, e.g.,  by quasi-Hamiltonian reduction
or by the approach of Fock and Rosly \cite{FR}. See also Section 6 in \cite{Res2}, and references therein.
Our systems should be related to real forms of these holomorphic systems.

We gave a linear algebraic algorithm for obtaining the solutions of our spin RS type models,
and advanced some arguments in favour of their degenerate integrability.
It could be
interesting to further explore their integrability properties.
The ultimate aim would be to exhibit the required number of (real analytic)
integrals of motion in explicit form on every symplectic strata \cite{SL,Sn} of $\cM_\red$.

The present work can be  generalized in such a way to obtain Poisson-Lie
analogues of the spin Sutherland models based on the cotangent lifts
of twisted conjugations acting on $G$ \cite{FP1}. Moreover, it should be
possible to use our approach as the starting point
for developing the quantum mechanics of the spin RS type systems.
For this, one should generalize the quantization of the spin Sutherland
models that proceeds via quantum Hamiltonian reduction and
related methods of harmonic analysis \cite{EFK,FP2}.

Another organic continuation of the present  work, whose implementation
is in progress in collaboration with I.~Marshall,  consists in deriving Poisson-Lie analogues
of the real, trigonometric Gibbons--Hermsen model \cite{GH}.
For this purpose, one needs to reduce a phase space of the form
\be
\GL(n,\C) \times \C^n \times \cdots \times \C^n
\ee
with an arbitrary number, $k\geq 2$,  of copies of $\C^n$. Here, $\GL(n,\C)$ is
the Heisenberg double of the Poisson-Lie group $\UN(n)$, and $\C^n$
carries a symplectic structure that enjoys $\UN(n)$ Poisson-Lie symmetry \cite{Zak}.
The pertinent moment map constraint can be written as
\be
\Lambda_L \Lambda_R \Lambda_1^{\C^n}\Lambda_2^{\C^n} \cdots \Lambda_k^{\C^n} = e^\gamma \1_n,
\ee
where $\gamma$ is a non-zero real constant, and the $\Lambda_i^{\C^n}$ ($i=1,\dots ,k)$ are
Poisson-Lie moment maps on
the independent copies of $\C^n$.
At the level of the  equations of motion,
we have proved that
this construction yields the trigonometric real form of the spin RS systems
studied earlier in \cite{AF,KZ,Rag},  but still have not completed
the description of the resulting reduced Poisson brackets.
For recent derivations of the Poisson structure of the Krichever--Zabrodin type trigonometric spin RS models
in the complex holomorphic case,
we refer to the  papers   \cite{CF2, AO}.  Our construction should lead to a real form of the Poisson structure
obtained in these papers.

\bigskip

\begin{acknowledgements}
 I wish  to thank  O. Chalykh, M. Fairon, T.F. G\"orbe, B. Jovanovic, I. Marshall, B.G. Pusztai  and
N. Reshetikhin for useful discussions.
I also thank  A. Alekseev, G. G\' evay and R. Sjamaar for suggesting
the train of thought described in  Remark 4.5.
 This work was supported  by the
Hungarian Scientific Research Fund (OTKA) under
 K-111697.
\end{acknowledgements}

\end{document}